\documentclass[11pt,a4paper]{article}

\usepackage{amsthm,amssymb,amsmath,amsfonts}
\usepackage{thmtools}
\usepackage{thm-restate}
\usepackage{tabularx}
\usepackage[square,numbers]{natbib}
\usepackage[
    colorlinks=true,
    linkcolor=black,
    citecolor=black,
    urlcolor=black
]{hyperref}
\usepackage{pgfplots} 
\pgfplotsset{compat=1.18}

\newtheorem{lemma}{Lemma}

\newtheorem{claim}{Claim}

\newcommand{\problemdef}[3]{
	\begin{center}
	\begin{minipage}{0.95\columnwidth}
		\noindent
		\textsc{#1}
		\vspace{5pt}\\
		\setlength{\tabcolsep}{3pt}
		\begin{tabularx}{\textwidth}{@{}lX@{}}
			\textbf{Input:}     & #2 \\
			\textbf{Question:}  & #3
		\end{tabularx}
	\end{minipage}
	\end{center}
}

\title{Integer Quadratic Programming is W[1]-Hard Parameterized by the Number of Variables}

\author{Anton Herrmann \\ Technische Universität Berlin \\\texttt{a.herrmann@tu-berlin.de}}

\begin{document}

\maketitle

\begin{abstract}
We show that \textsc{Integer Quadratic Programming} is W[1]-hard parameterized by the number of variables.
Thus, under standard complexity assumptions, \textsc{Integer Quadratic Programming} cannot be solved in $f(n) \cdot |\mathcal{I}|^{O(1)}$ time for any computable function~$f$ where $|\mathcal{I}|$ is the size of the encoding and $n$ is the number of variables.
\end{abstract}

\paragraph*{AI Disclosure:} The construction was found by GPT-5.6 Sol.
The author is responsible for the write-up and the correctness.

 
\section{Introduction}
In \textsc{Linear Programming (LP)} the goal is the optimization of a linear objective function subject to linear constraints.
It is a well-known result in mathematical optimization that this can be done in polynomial time \cite{KHACHIYAN198053}.
However, if the problem is restricted to integer values, then it becomes NP-hard \cite{Karp72}.
Formally, this restricted version is called \textsc{Integer Linear Programming (ILP)} and can be formalized as follows:
\begin{align*}
    \min \ & c^Tx \\
    \text{s.t. } & Ax \leq b \\
    & x \in \mathbb Z^n,
\end{align*}
where $A$ is an $m \times n$ integer matrix and $c,b$ are integer vectors of size $n$ and $m$ respectively.
Here, $c^Tx$ is the objective function and $Ax \leq b$ are the linear constraints.

Despite its NP-hardness, \textsc{Integer Linear Programming} is a very useful tool for parameterized algorithms.
In parameterized complexity, a problem is called fixed-parameter tractable with respect to some parameter $k$ if it can be solved in~$f(k)\cdot |\mathcal{ I}|^{O(1)}$ time for some computable function~$f$ where $|\mathcal{I}|$ is the encoding size of the input.
By a famous result of Lenstra \cite{Lenstra83}, \textsc{Integer Linear Programming} is fixed-parameter tractable with respect to the number of variables asa parameter.
Consequently, it is possible to show that some problem $A$ is fixed-parameter tractable with respect to some parameter $k$, by modeling the problem as an ILP in which the number of variables is bounded by a function of $k$.
This approach has been extensively used in the area of parameterized algorithms (see for example \cite{BoehmerK24,BredereckFNT21, CyganFKLMPPS15,GavenciakKK22, GimaHKKO22, KnopKMT19}).

Besides the application of Lenstra's algorithm, there has been a lot of research on improving the running time and obtaining similar results for generalizations of \textsc{Integer Linear Programming} \cite{Dadush12, FrankT87, HEINZ2005543, HildebrandK13, Kannan87, ReisR23}.
We refer to \cite{GavenciakKK22} for a recent overview of these results and their applications in parameterized complexity.
\newline

In this work, we consider the more general problem \textsc{Integer Quadratic Programming (IQP)}.
Here, the goal is to optimize a quadratic objective function subject to linear constraints.
Formally, the problem is defined as follows:
\begin{align*}
    \min \ & x^TQx + c^Tx \\
    \text{s.t. } & Ax \leq b, \\
    & x \in \mathbb Z^n,
\end{align*}
where $Q$ is a symmetric $n \times n$ integer matrix, $A$ is an $m \times n$ integer matrix and $c, b$ are integer vectors of size $n$ and $m$ respectively.
In the decision version of the problem we are given an additional integer $d$ and ask whether there exists an integer vector $x$ satisfying
\begin{align*}
    x^TQx + c^Tx &\leq d,\\
    Ax &\leq b.
\end{align*}
Lokshtanov \cite{Lokshtanov15} and Zemmer \cite{zemmer2017integer} showed that \textsc{Integer Quadratic Programming} is solvable in polynomial time if the number of variables plus the largest absolute value $\alpha$ in the matrices $Q$ and $A$ are fixed.
In particular, Lokshtanov's algorithm runs explicitly in $f(n + \alpha) \cdot |\mathcal{I}|^{O(1)}$ time.
Similarly to Lenstra's algorithm for ILP, this can be used as a tool to show fixed-parameter tractability of other problems by modeling them as IQPs~\cite{GavenciakKK22,Hanaka23,HlinenyS19,Lokshtanov15}.
Recently, Ari and Hildebrand \cite{ari2026} improved the dependency on $n +\alpha$ and applied their result to obtain improved parameterized algorithms for some graph problems.
Further work on the parameterized complexity of IQP was done by Eiben et al. \cite{EibenGKO19} who studied the problem under explicit restrictions on the domain or coefficients and under structural restrictions on variable interactions.
\newline

In this paper we negatively answer the question from Lokshtanov \cite{Lokshtanov15}, whether \textsc{Integer Quadratic Programming} is fixed-parameter tractable parameterized by the number of variables alone:
We show that IQP cannot be solved in $f(n) \cdot |\mathcal{I}|^{O(1)}$ time under standard complexity assumptions.

\paragraph*{Notation.}
For two integers $a\leq b$ we denote $[a,b] = \{a, a+ 1, \dots, b\}$ and $[a] = [1,a]$.
For a graph $G$, we denote by $V(G)$ the set of vertices and by $E(G)$ the set of edges.
An \emph{independent set} in $G$ is a set $S \subseteq V(G)$ of vertices such that $\{u,v\} \notin E(G)$ for all $u,v \in S$.

\paragraph{Parameterized Complexity.}
We use the definitions as stated in the book by~\citet{CyganFKLMPPS15} and refer to it for further details.
A \emph{parameterized problem} is a language $L \subseteq \Sigma^* \times \mathbb{N}$, where $\Sigma$ is a fixed, finite alphabet.
For an instance $(x,k) \in \Sigma^* \times \mathbb{N}$, the number $k$ is called the \emph{parameter}.
A parameterized problem is fixed-parameter tractable if it can be solved in~$f(k)\cdot |x|^{O(1)}$ time for some computable function~$f$.
If a parameterized problem is~$W[1]$-hard, then it is considered unlikely to be fixed-parameter tractable.
Let $A,B\subseteq \Sigma^* \times \mathbb{N}$ be two parameterized problems.
A \emph{parameterized reduction} from $A$ to $B$ is an algorithm that, given an instance $(x,k)$ of $A$, outputs an instance $(x',k')$ of $B$ such that
\begin{itemize}
    \item $(x,k)$ is a yes-instance of $A$ if and only if $(x',k')$ is a yes-instance of $B$,
    \item $k' \leq g(k)$ for some computable function $g$, and
    \item the running time is $f(k) \cdot |x|^{O(1)}$ for some computable function $f$.
\end{itemize}
In particular, if $A$ is W[1]-hard and there exists a parameterized reduction to $B$, then $B$ is also W[1]-hard.

\section{Hardness Proof}

We provide a parameterized reduction from \textsc{Independent Set (IS)} to IQP, where the former problem is defined as follows:

\problemdef{Independent Set}
{A graph $G$ and an integer~$k\in \mathbb N$.}
{Is there an independent set of size $k$ in $G$?}

\textsc{Independent Set (IS)} is known to be W[1]-hard parameterized by the solution size $k$ \cite{CyganFKLMPPS15}.
Given an IS instance $(G,k)$, we now construct in polynomial time an IQP instance with $2k$ variables which attains a minimum value of $\leq 0$ if and only if the graph $G$ contains an independent set of size~$k$.
This implies the main theorem.

\begin{restatable}{theorem}{mainthm}
    \label{thm}
    IQP is W[1]-hard parameterized by the number of variables.
\end{restatable}

From now on, we assume that the vertices of $G$ are given by the numbers $\{1, \dots, n\}$.
For each $i \in [k]$ we introduce two variables $x_i$ and $y_i$.
Intuitively, the values of $x_i$ and $y_i$ determine the $i$-th vertex in the independent set.
To be precise, $(x_i, y_i) = (t,t^2)$ corresponds to $t \in [n]$ being selected as the $i$-th vertex of the independent set.

The objective function of our IQP is defined as follows:
\begin{align*}
    f(x_1, \dots, x_k, y_1, \dots, y_k) = \sum_{i=1}^k (y_i - x_i^2).
\end{align*}
Clearly, the function attains value zero if each pair of variables $(x_i,y_i)$ corresponds to a correct choice of a vertex, that is, $(x_i,y_i) = (t,t^2)$ for some $t \in [n]$.
As a first step we introduce constraints that ensure $y_i \geq x_i^2$ for all $i \in [k]$ and afterwards show that the value zero is attained if and only if the values of the variables correspond to a selection of vertices.

We remark that the objective function $f$ is separable concave and therefore our hardness result even holds when restricted to separable concave functions.
Note, that this does not contradict the polynomial-time solvability for a constant number of variables in this case (which is mentioned by Del Pia \cite{Pia19}).

\paragraph{Choice Constraints.}

For each pair of variables~$(x_i,y_i), i \in [k]$, we introduce the following constraints:
\begin{align}
    1 &\leq x_i \leq n, \label{con:x_bounds} \\
    y_i &\geq (2r + 1)x_i - r(r+1), \hspace{0.5cm}  \text{for each } r \in [1,n]. \label{con:y_lower}
\end{align}

For each $r \in [n]$ we denote the function on the right-hand side of Constraint (\ref{con:y_lower}) by $h_r(x_i) := (2r + 1)x_i - r(r+1)$.
Intuitively, the linear function $h_r$ passes through the points $(r,r^2)$ and $(r+1, (r+1)^2)$ and therefore ensures that $y_i$ is at least $x_i^2$.
For a sketch see Figure \ref{fig:cons_y_value}.
\begin{figure}
    \centering
    \begin{tikzpicture}
        \begin{axis}[
            axis lines=middle,
            xlabel={$x$},
            ylabel={$y$},
            xmin=0, xmax=3.5,
            ymin=0, ymax=10,
            samples=100,
            width=9cm,
            height=7cm,
            xtick={1,2,3},
            ytick={1,4,9},
            clip=false
        ]
        
        \addplot[thick, domain=0:3.2] {x^2};
        
        \addplot[very thick, dashed, blue, domain=0.7:2.4] {3*x-2};
        
        \addplot[very thick, dashed, red, domain=1.7:3.3] {5*x-6};
        
        \addplot[
            only marks,
            mark=*,
            mark size=2pt
        ] coordinates {
            (1,1)
            (2,4)
            (3,9)
        };
        
        \node[below right] at (axis cs:1,1) {$(1,1)$};
        \node[above left] at (axis cs:2,4) {$(2,4)$};
        \node[above left] at (axis cs:3,9) {$(3,9)$};
        
        \end{axis}
    \end{tikzpicture}
    \caption{A sketch of the Constraint (\ref{con:y_lower}) with \textcolor{blue}{$r=1$} and \textcolor{red}{$r=2$}.
    These secants of the standard parabola $x^2$ guarantee $y_i \geq x_i^2$ for all $i \in [k]$.}
    \label{fig:cons_y_value}
\end{figure}
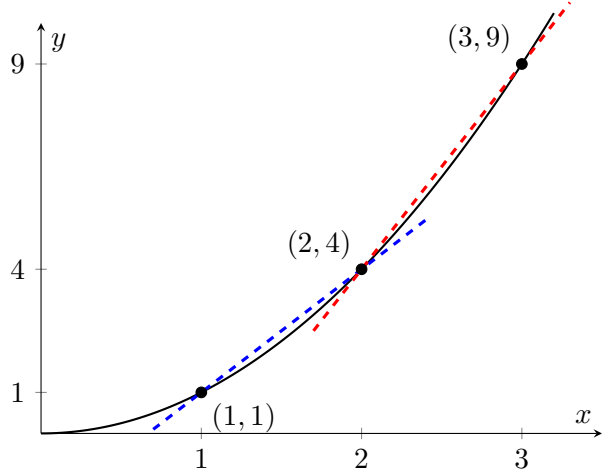
Formally this is proven in the next claim.

\begin{claim}
    \label{claim:choice_forward}
    Constraints (\ref{con:x_bounds}) and (\ref{con:y_lower}) ensure $y_i \geq x_i^2$.
\end{claim}

\begin{proof}
    By Constraint (\ref{con:x_bounds}) $x_i$ takes a value~$t$ between $1$ and $n$.
    Subsequently, Constraint (\ref{con:y_lower}) with $r = t = x_i$ yields:
    \begin{align*}
        y_i \geq h_{r}(x_i) = h_t(t) = (2t + 1)t - t(t + 1) = t^2 = x_i^2.
    \end{align*}
\end{proof}

It follows that for each $i \in [k]$ the summand $y_i - x_i^2$ in the objective function $f$ is nonnegative.
Thus, $f$ attains the value zero if and only if for all~$i\in [k]$ we have $(x_i, y_i) = (t,t^2)$ for some~$t \in [n]$.

Conversely, we now prove that the Constraints (\ref{con:x_bounds}) and (\ref{con:y_lower}) are satisfied for $i \in [k]$ if $(x_i, y_i) = (t,t^2)$ for some~$t \in [n]$.

\begin{claim}
    \label{claim:choice_backward}
    Let $(x_i, y_i) = (t,t^2)$ for some $t \in [n]$.
    Then Constraints (\ref{con:x_bounds}) and (\ref{con:y_lower}) are satisfied.
\end{claim}

\begin{proof}
    Constraint (\ref{con:x_bounds}) is clearly satisfied.
    For Constraint (\ref{con:y_lower}), let $r \in [1,n]$.
    Then 
    \begin{align*}
        t^2 - h_r(t) &= t^2 - \big( (2r+1)t - r(r+1) \big )\\
        & = t^2 - 2rt - t + r^2 +r = (t-r)(t-r-1),
    \end{align*}
    which is always nonnegative since the product of two consecutive integers is always nonnegative.
    Hence, $y_i = t^2 \geq h_r(t) = h_r(x_i)$ for all $r \in [1,n]$.
\end{proof}

We summarize the two claims from above in the next lemma.

\begin{lemma}
    \label{lem:value_zero}
    The minimum of the objective function $f$ subject to the Constraints (\ref{con:x_bounds}) and (\ref{con:y_lower}) is zero and is attained if and only if each pair $(x_i,y_i)$ is of the form $(t,t^2)$ with $t \in [1,n]$.
\end{lemma}

\paragraph{Independent Set Constraints.}
Next, we add constraints which force the values of $x_1, \dots, x_k$ to induce an independent set in $G$ if $y_i = x_i^2$ for all $i \in [k]$. 
In particular, we want to ensure $x_i \neq x_j$ and $\{x_i,x_j\} \notin E(G)$ for~$i \neq j$.

For every vertex $v \in [n]$ we define the linear function
\begin{align*}
    \phi_v(x,y) := 2vx - y.
\end{align*}
From a geometrical point of view, if $x$ and $y=x^2$ are fixed and $v$ is seen as the parameter, then the function $\phi_v$ is the tangent of the standard parabola through the point $(x,x^2)$.
In particular, for $t\in [n]$ we have
\begin{align*}
    \phi_v(t,t^2) := 2vt - t^2 = v^2 - (t-v)^2 \leq v^2,
\end{align*}
and equality $\phi_v(t,t^2) = v^2$ if and only if $v = t$.
With this in hand, we introduce for each pair $i\neq j \in [k]$ and each vertex $v \in [n]$ the constraint
\begin{align}
    \phi_v(x_i, y_i) + \phi_v(x_j,y_j) \leq 2v^2 - 1. \label{con:identity}
\end{align}
If $y_i = x_i^2$ and $y_j = x_j^2$, then this constraint is satisfied if and only if $x_i \neq v$ or $x_j \neq v$ as otherwise both summands would equal $v^2$ by the observation made above.

By the same idea, we add for each pair $i\neq j \in [k]$ and each pair of vertices~$(u,v) \in [n]^2$ with $\{u,v\} \in E(G)$ the constraint
\begin{align}
    \phi_u(x_i, y_i) + \phi_v(x_j,y_j) \leq v^2 + u^2 - 1. \label{con:edge}
\end{align}
If $y_i = x_i^2$ and $y_j = x_j^2$, then this constraint is satisfied if and only if $(x_i, x_j) \neq (u,v)$ as otherwise $\phi_u(x_i, y_i)$ and $\phi_v(x_j,y_j)$ would equal $u^2$ and $v^2$ respectively.

\paragraph{Deriving the Main Theorem.}
Finally, we prove the main theorem.

\mainthm*

\begin{proof}
    The constructed IQP has~$2k$ variables $x_1, \dots,x_k,y_1, \dots,y_k$ and the number of constraints is bounded by $O(k^2n^2)$.
    Moreover, the largest coefficient is of order $O(n^2)$ and therefore the reduction runs in polynomial time.
    It remains to prove that $G$ contains an independent set of size $k$ if and only if the optimal value is at most $0$.
    
    ``$\Rightarrow$'': Suppose $G$ contains an independent set $S = \{t_1, \dots, t_k\} \subseteq [n]$ of size $k$.
    We define
    \begin{align*}
        x_i := t_i \text{ and } y_i := t_i^2 \hspace{0.3cm} \text{for all } i\in [k].
    \end{align*}
    Clearly, $f(x_1, \dots, x_k, y_1, \dots, y_k) = 0$ and by Lemma \ref{lem:value_zero} we know that the Constraints (\ref{con:x_bounds}) and (\ref{con:y_lower}) are fulfilled.
    For the Constraints (\ref{con:identity}) and (\ref{con:edge}) we note that for all $i \neq j \in [k]$ we have $t_i \neq t_j$ and $\{t_i,t_j\} \notin E(G)$.
    Thus, by construction of the functions $\phi_v, v \in [n]$ the constraints are satisfied.
    
    ``$\Leftarrow$'': Suppose the IQP attains a value of $\leq0$ at~$(x_1, \dots, x_k, y_1, \dots, y_k)$.
    By Lemma \ref{lem:value_zero} we know that~$y_i = x_i^2$ and~$x_i \in [n]$ for all~$i \in [k]$.
    Now consider the corresponding set of vertices~$\{x_1, \dots, x_k\} \subseteq [n]$.
    Since~$y_i = x_i^2$ and~$y_j = x_j^2$, the Constraints (\ref{con:identity}) and (\ref{con:edge}) imply that $x_i \neq x_j$ and $\{x_i,x_j\} \notin E(G)$ for all $i\neq j \in [k]$.
    Hence, $\{x_1, \dots, x_k\}$ is an independent set of size $k$ in $G$.
\end{proof}

\section{Conclusion}
We provided a short and self-contained reduction showing that IQP is W[1]-hard with respect to the number of variables.
We conclude by recalling two open questions from the literature which are related to our result:
\begin{itemize}
    \item Is IQP fixed-parameter tractable parameterized by the number of variables plus the number of constraints? \cite{Lokshtanov15}
    \item Is IQP solvable in polynomial time if the number of variables is fixed \cite{PiaDM17}?  For two variables this is true \cite{PiaW14}.
\end{itemize}

\bibliographystyle{plainnat}
\bibliography{ref}
\end{document}